\documentclass[conference,compsoc]{IEEEtran}

\ifCLASSOPTIONcompsoc
  \usepackage[nocompress]{cite}
\else
  \usepackage{cite}
\fi

\usepackage{amssymb}
\usepackage{amsthm}
\usepackage{amsmath}
\usepackage{algorithm}
\usepackage{algpseudocode}
\usepackage{pifont}

\newtheorem{theorem}{Theorem}

\newtheorem{definition}{Definition}
\newtheorem{assumption}{Assumption}
\newtheorem{corollary}{Corollary}

\begin{document}

\title{Rollerchain, a Blockchain With Safely Pruneable Full Blocks}

\author{\IEEEauthorblockN{Alexander Chepurnoy}
\IEEEauthorblockA{IOHK Research\\
Sestroretsk, Russia\\
Email: alex.chepurnoy@iohk.io}
\and
\IEEEauthorblockN{Mario Larangeira}
\IEEEauthorblockA{IOHK Research\\
Tokyo, Japan\\
Email: mario.larangeira@iohk.io}
\and
\IEEEauthorblockN{Alexander Ojiganov}
\IEEEauthorblockA{Saint-Petersburg National Research University of\\
 Information Technologies, Mechanics and Optics\\
Saint-Petersburg, Russia}}

\maketitle

\onecolumn

\begin{abstract}

Bitcoin~\cite{Nakamoto2008} is the first successful decentralized global digital cash system. Its mining process requires intense computational resources, therefore its usefulness remains a disputable topic. We aim to solve three problems with Bitcoin and other blockchain systems of today by repurposing their work. First, space to store a blockchain is growing linearly with number of transactions. Second, a honest node is forced to be irrational regarding storing full blocks by a way implementations are done. Third, 
a trustless bootstrapping process for a new node involves downloading and processing all the transactions ever written into a blockchain.  

In this paper we present a new consensus protocol for Bitcoin-like peer-to-peer systems where a right to generate a block is given to a party providing non-interactive proofs of storing a subset of the past state snapshots. Unlike the blockchain systems in use today, a network using our protocol is safe if the nodes prune full blocks not needed for mining. 

We extend the GKL model~\cite{garay2015bitcoin} to describe our Proof-of-Work scheme and a transactional model modifications needed for it. We provide a detailed analysis of our protocol and proofs of its security.

\end{abstract}

\IEEEpeerreviewmaketitle

\section{Introduction}
\label{intr_section}

The Bitcoin whitepaper~\cite{Nakamoto2008} defines a way a common ledger could be maintained within a trustless peer-to-peer network by using moderately hard computational puzzles~\cite{milleranonymous} and the blockchain datastructure. Since then a lot of research has been done about Bitcoin and alternative systems. Nevertheless, there are still many open problems in the field, and performance is one of the most crucial~\cite{cromanscaling}. Another threat to the blockchain-based systems of today is the lack of reward for any activity supporting the network other than the block generator self-election via finding a partial hash collision. In particular, there is no reason for the nodes in the network to store all the blocks since the genesis until the few newer ones. On other side, it is not known how to build a safe network if its participants are going to be rational about storing the full blocks.

In this paper we present a consensus protocol alternative to the Proof-of-Work used in Bitcoin. In our protocol a participant is rewarded for archiving few state snapshots amongst \(n\) states a network aims to store collectively. That is, if a miner is storing a state snapshot for height \(h\)\footnote{Here we use the regular meaning for the word {\it height}, that is it is the number of blocks of the blockchain from its genesis to the given block}, then when a new block appears she will need to replace the snapshot with a new one corresponding to height \(h+1\). Thus a miner needs to store some number of full blocks also.

\subsection{The Consensus and the Mining Lottery}

The Bitcoin blockchain is generated in a worldwide peer-to-peer network without any central authority. In such an environment, the next block determination process requires a protection against Sybil attacks~\cite{douceur2002sybil}. 

The Proof-of-Work mining process~\cite{Nakamoto2008} eliminates Sybil attacks and also enforces a rational miner to choose a single version of the history out of possible options. 

The Proof-of-Work mining process could be seen as a lottery as described in~\cite{miller2014permacoin}. A mining software iterates over changeable block field values until it finds a solution, that is, a block whose header satisfies the predicate: \( hash(blockheader) < difficulty \)\footnote{In general ``hash'' means a cryptographic hash function, whereas ``difficulty'' is defined analogously as in the Bitcoin protocol jargon.}, where \textit{difficulty} is a publicly known value. The iterations could not be precomputed because of the epoch-dependent\footnote{Here, as in~\cite{miller2014permacoin}, we consider the {\it epoch} as the interval between the discovery of blocks.} unique value which is known to all the participants and stored into a blockheader. 

\subsection{The Blockchain Storing Rationale}
\label{rationale}

In the Bitcoin Core, the reference implementation of the Bitcoin protocol~\cite{bitcoind}, a node stores all the blocks with all the transactions since the genesis block. In addition, the current state snapshot is also stored in the form of unspent transaction outputs set. 

The main benefit for the network in storing all the full blocks is that the new nodes have to download them since the genesis and re-validate all the transactions ever included in the blocks. Therefore, in order for the blockchain holder to be useful for the network, i.e., perform an  altruistic activity, it has to  store all the full blocks. 

A \textit{rational} participant can still have all the full node merits without  keeping all the blocks except the last few ones which are needed to handle possible forks~(a probability for a blockchain suffix to be reversed is going down exponentially with a suffix length). To be rational it is needed to switch to an implementation of the Bitcoin protocol which does not hold the unnecessary data~(such an implementation of a full node does not exist to the best of our knowledge). 

However, if all the nodes in the network are rational, \textit{tragedy of commons}~\cite{hardin1968commons} happens: nobody is storing full blocks except last few ones. The only way to get into the network then is to download the current state snapshot and trust it. 

\subsection{The Motivation For a New Protocol}
\label{section_motiv}

Next, we detail the four main points which our protocol addresses.

\subsubsection{Incentives to Keep a Full Node}

 The only rewarded activity in Bitcoin is the iteration over values for certain fields in a block header. If mining software includes transactions into the block, it needs to validate them, which can be carried by presenting the current state snapshot. Storing any number of full blocks is not needed in order to mine. Our particular goal is to develop a protocol providing incentives to run a node storing sufficiently enough number of full blocks for network viability, thus making the network safer in the long run. 

\subsubsection{Solving The Blockchain Storing Rationale Problem}

As the growth of a blockchain is not bounded, in the long run the network can survive only if storage and processing power of an ordinary computer grows not slower than requirements of a blockchain system. Storing enough blocks is the activity with no reward, so eventually most of the full nodes will purge blocks from their disks except the last ones having non-negligible probability to be rolled back. This means a system where blocks are needed but not stored due to a practical impossibility. Our protocol rewards the collectively storage of number of blocks sufficiently enough for network safety. 

\subsubsection{A Prunable Blockchain}

A solution to the storing rationale problem would be a protocol for which an archive of some number of blocks and states are required in order to mine new blocks. This number \(n\) could be sufficiently large to assume any rollback caused by a fork could be of negligible depth in comparison with it. For example, \(n = 10,000\) blocks gives approximately 70 days of history to be stored if a block is generated every \(10\) minutes on average. Any rollback possible is of negligible depth in comparison with that. 

In such a system we can assume that the overwhelming majority of the nodes (except tracking tools and other special applications) to be rational so to remove blocks not needed for mining having a freedom of choice given. A rational node has predictable storage resources consumption. A network of rational nodes is also sending less blocks, and bandwidth being saved could be repurposed for other tasks improving overall system performance.

\subsubsection{Fast Trustless Bootstrap}

A new full node in Bitcoin needs to download and process all the blocks since the genesis block. This results into an unreasonably long and resource-consumptive initial processing phase. In order to reduce the burn of this phase, a trusted Bitcoin state snapshot generated by a notable community member could be downloaded~\cite{snapshot}. This approach reinstates trust-related issues solved by Bitcoin. We want to provide fast, trustless and safe bootstrapping with the help of our consensus protocol.

\subsection{The Bitcoin Backbone Protocol}
\label{gkl_general}

There are few models describing Bitcoin consensus protocols. Simpler ones from~\cite{miller2014permacoin, milleranonymous} introduce a notion of a \textit{ticket}, which is the one iteration of a Proof-of-Work function. Comprehensive GKL model from~\cite{garay2015bitcoin} derives properties of the Bitcoin protocol as well as ledgers built on top of it while not using a notion of ticket. We will use the latter work to reason about our protocol. 

The GKL model~(also the Bitcoin Backbone protocol) relies on the standard choice for the description of distributed models~\cite{Canetti:2001} which we now describe.

 Every party  $P$ in the network is modeled as an Interactive Turing Machine (ITM) which has access to two tapes named $INPUT()$ and $RECEIVE()$, the {\it input} and {\it communication tapes}. The interaction of the parties are controlled by the environment entity $\mathcal{Z}$: an ITM which provides the contents of the $INPUT()$ tape of each party. It also defines the {\it rounds} of the system. At each round the parties are allowed to write and read from its tapes as well as perform computation. We also assume the existence of an operation $BROADCAST$ which allows the parties to send messages atomically to all other parties. Furthermore, each party is allowed to execute at most $q$ hash queries to a regular cryptographic hash function denoted by $hash$.  

\paragraph{\bf The adversarial and messaging framework} 
The adversary $\mathcal{A}$ is also an ITM. We assume it is ``adaptive'' which means it can corrupt honest parties $P_{i}$ which are performing the protocol, say  $\Pi$. It also has access to at most $q$ queries to the hash function, however which can be added to the queries from the corrupted parties. 

The messages exchanged between the parties $P_{i}$ can be intercepted by $\mathcal{A}$. Moreover, the adversary can change the origin of the message, however it cannot delay the delivery of it nor change its contents (exception is the origin, as stated earlier). Hence the model guarantees the delivery of the messages on the next round. In other words, there is no {\it delay} on the delivery of the messages. This feature characterizes the model as {\it synchronous}, as opposite to the {\it asynchronous model} when the adversary can add a delay on the delivery of messages. Given the  upper bound $q$ hash queries for each party, we denote the model $q${\it -bounded synchronous model}.

\paragraph{\bf The Execution of $\mathbf{\Pi}$}
The execution of the protocol $\Pi$ is captured by the view of the environment $\mathcal{Z}$ is denoted by $\mathsf{VIEW}_{\Pi,\mathcal{A},\mathcal{Z}}(k,q,z)$ which is the concatenation of the views of each  $P_{i}$ performing the protocol $\Pi$. That is $p$ random variable ensembles\footnote{Here we consider an unknown, for the parties, number of $p$ participants, where $p$ is a fixed value.} $\{\mathsf{VIEW}_{\Pi,\mathcal{A},\mathcal{Z}}^{P_{i}}(k,q,z)\}_{k\in\mathbb{N},z\in\{0,1\}^\ast}$, for the security parameter $k$ and auxiliary input $z$.

For a concrete ledger to be cast into the GKL Model it is necessary to provide a construction for three functions. They are:
\begin{itemize}
\item Content Validation Predicate  $V(\cdot)$
\item Chain Reading Function $R(\cdot)$
\item Input Contribution Function $I(\cdot)$
\end{itemize}

We will define the three functions to some degree needed for our Proof-of-Work scheme. Then concrete blockchain system~(Bitcoin, Namecoin, Ethereum and so on) can extend our definitions to complete its design. 

\subsection{Rollerchain}

We shape some properties of ledger semantics and Proof-of-Work scheme built with respect to them. The framework allows to achieve goals claimed in the section~\ref{section_motiv} via rewarding miners to store collectively a rolling window of state snapshots and full blocks~(thus the name \textit{Rollerchain}).

\section{Transactional Model}
\label{tx_model}

Our consensus protocol requires for a ledger with some properties. In order to formally define the properties we are extending the Bitcoin backbone protocol described in general in section~\ref{gkl_general}. Unlike Bitcoin, we are adding an authenticating value for a whole state to a block. There are discussions in the Bitcoin community about implementing that~(the earliest found discussion was started by Andrew Miller back in 2012~\cite{utxoset}), but no concrete plans exist at the moment to the best of our knowledge.

\subsection{The Blockchain and The State}

A blockchain could be seen as a linked list where an element (a \textit{block}) is a tuple \((\Delta S_i, \Delta C_i)\), where \(\Delta S_i\) is a \textit{transactional state modifier} and \(\Delta C_i\) is a \textit{consensus state modifier}. The tuple of state modifiers could be applied to a state which results another state whenever the modifiers are valid. We denote a state modifier application by \(\diamond\), then \((S_i,C_i) = (S_{i-1},C_{i-1}) \diamond (\Delta S_i, \Delta C_i) \). Every network participant knows the predefined  transactional and consensus state \((S_1, C_1)\) resulted from the \textit{genesis block}. Then each participant knows exactly the same state \((S_i,C_i) = (((S_1, C_1) \diamond (\Delta S_2, \Delta C_2)) \diamond \dots \diamond (\Delta S_i, \Delta C_i)) \) if all the blocks are the same. The order of modifiers is defined via immutable  link from each modifier to the previous one, where the first modifier in the history must be linked to the genesis state. For a block \((\Delta S_i, \Delta C_i)\) we denote by \(i\) the value for the \textit{height} of the block.

The consensus state modifier changes the rules on the block validation which are not related to the transactions stored in it (for example, it contains the \textit{difficulty} value  which is explained in the next sections). A transactional state modifier being, atomic in terms of its application, itself contains a sequence of transactions. Furthermore, given the height \(i\), we denote by \textit{state snapshot} or just \textit{state} the transactional state \(S_i\).

\subsection{The Fixed State Representation}
\label{fixed_rep}

A state representation is not fixed by the Bitcoin protocol. A full node implementation usually stores a set of unspent outputs and also some node-specific additional information. By applying valid transactions from a new block, a node software takes unspent outputs out of the set and puts there outputs from the transactions in the new block~\cite{Nakamoto2008}. 

Abstracting the Bitcoin-like model, a state could be represented as a set of \textit{closed boxes} of size \(n_S\). Each box has a value associated with it. A transaction contains openers for \(n_k\) boxes and also creates \(n_b\) new closed boxes. The resulting state set has the size of \(n_S-n_k+n_b\) after applying the transaction to it. 

Each box has some unique identifier \(id(box)\) thus the state could be represented as dictionary \((id(box) \rightarrow box)\). We require the dictionary to be \textit{authenticated} and corresponding one-way digest to be included into a blockheader. Note that our construction of a state representation is fixed and is a part of the protocol, unlike Bitcoin.

We use the term \textit{box} and not \textit{output} because the latter is used not in all the blockchain systems. For example, Ethereum\cite{ethyp} is using notion of mutable \textit{accounts} instead of immutable outputs, but we still can get a box from an account and its state in order to build an explicit state using boxes.

\subsection{An Authenticated Dictionary}

We represent the state in the form of \(id(box) \rightarrow box\) correspondences, and an authenticated dictionary~\cite{nissim1998certificate, kocher1998certificate} is to be built upon them. Different authenticated dictionary implementations are known: sparse Merkle trees~\cite{sparsemerkle}, treaps~\cite{crosby2011authenticated}, skiplists~\cite{anagnostopoulos2001persistent}, balanced trees~\cite{nissim1998certificate}, tuple-based solutions~\cite{crosby2011authenticated}. We do not specify a concrete implementation for an authenticated dictionary but require the following properties to be hold:

\begin{itemize}
\item \textbf{Root authenticator.} A single fixed-size value commits the entire dictionary.
\item \textbf{Set-uniqueness.} A dictionary with given set of keys has a unique and canonical representation.
\item \textbf{Efficiency.} The proof returned for a lookup request should be has a size sublinear to dictionary size.
\item \textbf{Non-membership proofs mentioning set member ids.} Our protocol will generate uniform ids so in most cases it will be no element in the set with a key given. Thus we need for proofs of non-membership. As we are going to include a box into a blockheader, we require a non-membership proof to mention a member id or ids. We assume there is a function \(member(\pi)\) which extracts in a deterministic way an id of a member of the dictionary from a proof, whether it is a proof of membership or a proof of non-membership.
\end{itemize}

An authenticated dictionary must provide a support for following operations:

\begin{itemize}
\item{\(root(\mathcal{D})\)} calculates an authenticating value for a dictionary \(\mathcal{D}\).
\item{\(checkRoot(\mathcal{D}, a_\mathcal{D})\)} checks whether \(a_\mathcal{D}\) is a correct authenticating value for a dictionary \(\mathcal{D}\).
\item{\(generate(\mathcal{D}, i\d)\)} generates a proof \(\pi\) of (non-)~membership for an identifier \(id\) and dictionary \(\mathcal{D}\).
\item{\(checkPath(\mathcal{D}, id, \pi)\)} checks whether a proof \(\pi\) is valid for an identifier \(id\) and dictionary \(\mathcal{D}\).
\item{\(member(\mathcal{D}, \pi)\)} returns deterministically defined element presenting in dictionary \(\mathcal{D}\) given a proof of (non-)~membership \(\pi\).
\end{itemize}

Note we define operation for a dictionary \(\mathcal{D}\), but it also possible to use a set of uniquely identifiable objects instead in an every operation. In this case we first extract an identifier for each object getting a dictionary as a result of this transformation, then we apply an operation to the dictionary. For example, we will write \(root(\tau)\), where \(\tau\) is a set of transactions minding the transformation to be done before the operation.

\subsection{Block Header}

Our idea to reduce storage requirements based on a notion of a \textit{block header}: 

\begin{definition}
A block header contains parts of a block enough to check its authenticity and whether a valid amount of work has been spent to generate it. In Rollerchain, block header is $\langle s, t, root(\mathcal{S}), root(\tau)\rangle$.
\end{definition}

In order to build a safe system we need full nodes to store all the block headers since genesis, thus the following assumption: 

\begin{assumption}
Throughout the paper we assume a rational full node can tolerate storing all the block headers since genesis. In the same time it prunes full blocks not needed for selfish purposes anymore to just block headers.
\end{assumption}

We argue the assumption is reasonable. As of August, 2016, a block header in Bitcoin is about just 80 bytes while a full block is about 1 megabyte. For 1 million block headers~(about 19 years of Bitcoin history), block headers fit into 80 megabytes while full blocks will consume 1 terabyte of disk space. 

\subsection{Refined Transactional Ledger Model}

The protocol parties, called \textit{miners}, process sequences of transactions \(\tau = tx_1 \dots tx_e\). A transaction contains identifiers of boxes to remove from a state along with openers and boxes to append: \(tx = \langle (id(box_{r_1}) \rightarrow opener_1, \dots, id(box_{r_k}) \rightarrow opener_k), (id(box_{a_1}) \rightarrow box_{a_1} , \dots , id(box_{a_j}) \rightarrow box_{a_j}) \rangle \). A transaction is valid against a state which is a set of boxes if it removes boxes presenting in the state with valid openers and append boxes not presenting in the state. Transactions as well a state after applying them are supposed to be incorporated into their local chain \(\mathcal{C}\). The input inserted at each block of the chain \(\mathcal{C}\) is the whole state along with its integrity proof \(a_\mathcal{S} = root(\mathcal{S})\) and transactions along with the integrity proof for them \(a_{\tau} = root(\tau)\) \((\mathcal{S}, a_\mathcal{S}, \tau, a_{\tau})\). Thus, a chain \(\mathcal{C}\) contains the vector \(x_\mathcal{C} = \langle (\mathcal{S}_1, a_{\mathcal{S}_1}, \tau_1, a_{\tau_1}),\dots,(\mathcal{S}_m, a_{\mathcal{S}_m}, \tau_m, a_{\tau_m})\rangle\).

Next, we define functions $ValidateBlock$~(to check validity of \(a_\mathcal{S}, \tau, a_\tau\) values from a block against a state \(\mathcal{S}_p\) previous to the block) and $\diamond$~(to apply a set of transactions $\tau$ to a state $\mathcal{S}$ getting an updated state as result).

\begin{algorithm}[H]
\caption{Block validation function, parametrized with constant block reward value \(constReward\)}
\label{algo_validate_block}
\begin{algorithmic}[1]
\Function{ValidateBlock}{ $\mathcal{S}_p$, $a_\mathcal{S}$, $\tau$, $a_{\tau}$ } 

\If {\(\tau\) is empty}
\Return false
\EndIf

\State \( tx_{coinbase} \gets head(\tau) \)

\If { $tx_{coinbase}$ creates more than 1 box or opens any box}
\Return false
\EndIf

\State $ box_{coinbase} \gets $ the only box of $ tx_{coinbase} $

\State $ value_{coinbase} \gets $ value of $ box_{coinbase} $

\State $ fee_{total} \gets 0 $

\State $ boxes_{new} \gets [box_{coinbase}] $

\State \(S^{\prime} \gets \mathcal{S}_p\)

\For{\textbf{each} transaction \(tx\) in \(tail(\tau)\)}
\State $ fee_{tx} \gets 0 $
\For{\textbf{each} box to remove \(box_r\) and its \(opener\) in \(tx\)}
\If{a a box with identifier \(id(box_r)\) is not in \(\mathcal{S}^{\prime}\) \textbf{or} $opener$ invalid}
\Return false
\EndIf

\State $ fee_{tx} \gets fee_{tx} - $ value of $box_r$

\State remove \(box_r\) from \(\mathcal{S}^{\prime}\)
\EndFor

\State $ fee_{tx} \gets fee_{tx} + $ sum of values of new boxes in \(tx\) 

\If{ \(fee_{tx} < 0\) }
\Return false
\EndIf

\State $ fee_{total} := fee_{total} + fee_{tx} $
\State add new boxes from $tx$ to $boxes_{new}$
\EndFor

\If{\(value_{coinbase} \neq fee_{total} + constReward\)}
\Return false
\EndIf

\State add \(boxes_{new}\) to \(\mathcal{S}^{\prime}\)

\State \Return \( checkRoot(\mathcal{S}^{\prime}, a_\mathcal{S}) \land checkRoot(\tau, a_{\tau}) \)

\EndFunction
\end{algorithmic}
\end{algorithm}

\begin{algorithm}[H]
\caption{Block application function \(\diamond\).}
\label{algo_apply_block}
\begin{algorithmic}[1]

\Function{ \(\diamond\) }{  $\mathcal{S}$, $\tau$  }

\State \(\mathcal{S}^{\prime} \gets \mathcal{S}\)

\For{\textbf{each} transaction \(t\) in \(\tau\)}

\For{\textbf{each} \(id(box_r)\) in t}
\State Remove \(box_r\) from \(\mathcal{S}^{\prime}\)
\EndFor

\For{\textbf{each} \(box_a\) in t}
\State Append \(box_a\) to \(\mathcal{S}^{\prime}\)
\EndFor

\EndFor

\State \Return \( \mathcal{S}^{\prime} \)

\EndFunction

\end{algorithmic}
\end{algorithm}

With the help of the functions defined previously we can now describe semantics of $V(\cdot), I(\cdot), R(\cdot)$ functions of the GKL model defined in the section~\ref{gkl_general}.

\begin{table}[H]
\caption{Box operations log protocol, built on the Bitcoin backbone.}
\label{operations}
\begin{tabular}{ |l|p{10cm}| }  
  \hline
  Content validation predicate \(V(\cdot)\) & \(V (\langle x_1, ..., x_m \rangle) \) is true if and only if for every \( x_i = (a_{\mathcal{S}_i}, \tau_i, a_{\tau_i}), i > 1, ValidateBlock(\mathcal{S}_{i-1}, a_{\mathcal{S}_i}, \tau_i, a_{\tau_i}) = true \), \(\mathcal{S}_i = \mathcal{S}_{i-1} \diamond \tau_i\), and \(\mathcal{S}_1\) is known valid genesis state. \\
  \hline
  Chain reading function \(R(\cdot)\) & If \(V (\langle x_1, ..., x_m \rangle) = True \), the value \(R(\mathcal{C})\) is equal to \(\langle x_1, ..., x_m \rangle\); undefined otherwise. \\
  \hline
  Input contribution function \(I(\cdot)\) &  \(I(\mathcal{C},round,Input())\) operates in the following way: if the input tape contains \((Insert,v)\), it parses \(v\) as a sequence of transactions and retains the largest subsequence \(\tau \preceq v\) that is valid with respect to the current state \(\mathcal{S}_c\) from a last block in \(\mathcal{C}\), and then \(\mathcal{S} = \mathcal{S}_c \diamond \tau, x = (root(\mathcal{S}), \tau, root(\tau)) \). \\
  \hline
\end{tabular}
\end{table}

A concrete blockchain system to be built on top of the Rollerchain needs to specify box semantics and authenticated dictionary implementation. 

\section{The Protocol}
\label{protocol}

Our protocol is designed to create an incentive for the miners to store collectively the last \(n\) states and blocks, where each miner stores at least \(k\) states and also \(\frac{k\cdot n}{k+1}\) blocks on average in order to generate a block. Before detailing our construction, it is convenient to introduce the notation that will be used from this point.

\subsection{The Notation}

We denote by \textit{hash}  a regular cryptographic hash function with a uniformly distributed output. Furthermore, given two strings \( z\) and \( w\) we denote by \( z || w \) the string which results from the concatenation of \( z\) and \( w\).

We assume that the mining rewards could be given to an owner of a public key \(pk\). Furthermore, we assume the existence of a signature scheme and every party owns a public key \(pk\). 

Function $last(\mathcal{C})$ returns last generated block from a chain $\mathcal{C}$.

\subsection{The Setup}

Consider a party wants to be a miner. In the first place she is generating her public key \(pk\) and choosing state snapshots to store based on the public key. The set of states is defined by the function \(ChooseSnapshots\) defined below. 

\begin{algorithm}[H]
\label{alg_snapshots}
\caption{Snapshots extraction function \(ChooseSnapshots\).}
\label{algo_choose_snapshots}
\begin{algorithmic}[1]

\Function{ ChooseSnapshots }{ $\mathcal{C}$, $pk$ }

\State $h_c \gets $ length of $\mathcal{C}$

\State $\mathcal{S}_{pk} \gets []$

\For{\textbf{each} \(i\) in \(1 \dots k\)}

\( h \gets (\mbox{\textit{hash}}(pk || i) \bmod{n}) + (h_c-n) \) 

\State add state corresponding to block $\mathcal{C}[h]$ to $\mathcal{S}_{pk}$ if \(h > 0\), genesis state otherwise (it could happen if \(h_c < n\))

\EndFor

\State \Return $\mathcal{S}_{pk}$

\EndFunction

\end{algorithmic}
\end{algorithm}

All the \(\mathcal{S}_{pk}\) values must be unique, otherwise no valid block could be generated. We enforce such a requirement to prevent malicious iteration over the public key space to find as much repeating states as possible. Thus it is not possible to mine with any public key.  When a new block at height \(h_c + 1\) arrives, Alice needs to recalculate at least \(k\) states. Thus Alice must store blocks since minimal height of \(\mathcal{S}_{pk}\) also.

[TODO: draw pic]

\subsection{The Ticket Generation}

Next, we describe a process of generating an object from \(k\) state snapshots defined by the $ChooseSnapshots$ function we name a \textit{ticket} following the Permacoin paper~\cite{miller2014permacoin}. Ticket consists of \(k\) boxes~(one from each of the states) along with proofs against state authenticating values. The function $GenTicket$ to generate a ticket is getting current blockchain $\mathcal{C}$, miner's public key $pk$, unpredictable seed value $s_t$ and a nonce $ctr$ to be increased on an each call with the same $(\mathcal{C}, s_t, pk)$ values.

\begin{algorithm}[H]
\caption{Ticket generation function \(GenTicket\).}
\label{algo_ticket_generation}
\begin{algorithmic}[1]

\Function{ GenTicket }{ $\mathcal{C}$, $s_t$, $pk$, $ctr$ }

\State $seed \gets ctr$
\State \(\mathcal{S}[1 \dots k] \gets ChooseSnapshots(\mathcal{C}, pk) \)

\For{\textbf{each} \(i\) in \(1 \dots k\)}
\State \( id_i \gets \mbox{\textit{hash}}(seed || pk || s_t) \)
\State \( \pi_i \gets generate(\mathcal{S}_i, id_i) \)
\State \( box_i \gets member(\mathcal{S}_i, \pi_i) \)
\State \( a_{\mathcal{S}_i} \gets root(\mathcal{S}_i) \)
\State \( seed \gets id_i \)
\EndFor

\State \(t \gets \langle (pk \rightarrow ctr), \forall i \in {1 \dots k}, (id_i \rightarrow (a_{\mathcal{S}_i}, \pi_i, box_i) \rangle \)
\State \(a_t = root(t)\)
\State \Return \( a_t \rightarrow t \) 
\EndFunction

\end{algorithmic}
\end{algorithm}

\subsection{Proof-of-Work Function}

We modify the Proof-of-Work function of the GKL model~\cite{garay2015bitcoin} by using \(GenTicket\) function defined above, and also we are explicitly adding miner's public key \(pk\) as an argument of the Proof-of-Work function.

\begin{algorithm}[H]
\caption{Rollerchain's Proof-of-Work function, parametrized by \(q, D\). The input is \((x; \mathcal{C}; pk)\).}
\label{algo_rc_pow}
\begin{algorithmic}[1]
\Function{RollerPow}{$x$, $\mathcal{C}$, $pk$}
\State $ \langle a_{\tau}, a_{\mathcal{S}}, \tau \rangle \gets x $
\If{$C = \varepsilon$}
	\State $s \gets 0$
\Else
	\State $ \langle s^\prime, {a_t}^\prime, a_{\tau}^\prime, a_{\mathcal{S}}^\prime, ctr^\prime \rangle \gets head(\mathcal{C}) $
	\State $s \gets hash(ctr^\prime, hash(s^\prime, t^\prime, a_{\tau}^\prime, a_{\mathcal{S}}^\prime) $
\EndIf
\State $ctr \gets 1 $
\State $B \gets \varepsilon$

\While{$(ctr \leq q)$}
\State $s_t \gets hash(s || a_\mathcal{S} || a_\tau)$
\State $\langle a_t, t \rangle \gets GenTicket(C, s_t, pk, ctr) $
\State $h \gets hash(s, a_t, a_{\tau}, a_{\mathcal{S}}) $
\If{$hash(ctr, h) < D$}
\State $B \gets \langle s, t, x, ctr \rangle$
\State \textbf{break}
\EndIf
\State $ctr \gets ctr + 1$
\EndWhile
\State $\mathcal{C} \gets \mathcal{C}B$
\State \Return $\mathcal{C}$
\EndFunction
\end{algorithmic}
\end{algorithm}

Note, in line 12 we are generating a seed value for a \(GenTicket\) procedure as $hash(s || a_\mathcal{S} || a_\tau)$. It is made to avoid an optimization when a miner could generate a ticket once and then iterate over $a_\tau$ and $a_\mathcal{S}$.

\subsection{Protocol Notes}

With our consensus protocol, block header becomes $\langle s, a_t, a_\tau, a_\mathcal{S}, ctr \rangle$, a full block consists of a block header plus full ticket $t$ and transaction set $\tau$.

In order to check block validity one needs, in addition to all the checks described before in this paper, one needs to replay $GenTicket()$ code with boxes and proofs of their authenticity given and $checkPath()$ instead of $generate()$ and then $member()$.

\section{Discussion of The Protocol}
\label{discussion}

In this section we analyze the properties of our protocol and also relate our construction with potential issues.

\subsection{The Bootstrapping Process}

In blockchain systems of today, a bootstrapping process for a new full node is as follows:

\begin{enumerate}
\item A node knows genesis state.
\item A node downloads all the full blocks and apply them. 
\end{enumerate}

We propose to use following light bootstrapping algorithm instead of the classic one:

\begin{enumerate}
\item A node knows genesis state.
\item A node downloads block headers, check a chain Proof-of-Work validity. 
\item A node asks peers for available states. 
\item A node downloads a state $\mathcal{S_i}$ for the height $i$ from available options.
\item A node downloads full blocks since $i$ and apply them. 
\end{enumerate}

Bitcoin at this moment has about 40 million unspent outputs and 160 million transactions, and a size of an output is by an order of magnitude smaller than a size of transaction. Thus the light bootstrapping allows to reduce network traffic during downloading chain prefix without transactions and also eliminate transactions validation. The latter now takes tens of hours on commodity hardware.

\subsection{Security Analysis}
\label{protocol_props}

There are two key questions about the security of the proposal we need to answer. In the first place, how different is our Proof-of-Work function from the classical one used in Bitcoin. In the second place, how secure is the light bootstrapping in comparison with full validation? 

\subsubsection{Proof-of-Work Equivalence}

We want to prove our RollerPow function could be used instead of BitcoinPow. Below is the BitcoinPow function from~\cite{garay2015bitcoin} with respect to our definition of $x$ argument and the only hashing function $hash$ to be used.

\begin{algorithm}[H]
\caption{The Bitcoin's proof of work function, parametrized by \(q, D\). The input is \((x; \mathcal{C})\).}
\label{algo_bitcoin_pow}
\begin{algorithmic}[1]
\Function{BitcoinPow}{x,C}
\State $ \langle a_{\tau}, a_{\mathcal{S}}, \tau \rangle \gets x $
\If{$\mathcal{C} = \varepsilon$}
	\State $s \gets 0$
\Else
	\State $ \langle s^\prime, a_{\tau}^\prime, a_{\mathcal{S}}^\prime, ctr^\prime \rangle \gets last(\mathcal{C}) $
	\State $s \gets hash(ctr^\prime, hash(s^\prime, a_{\tau}^\prime, a_{\mathcal{S}}^\prime)) $
\EndIf
\State $ctr \gets 1 $
\State $B \gets \varepsilon$
\State $h \gets hash(s, a_{\tau}, a_{\mathcal{S}}) $

\While{$(ctr \leq q)$}
\If{$hash(ctr,h) < D$}
\State $B \gets \langle s, x, ctr \rangle$
\State \textbf{break}
\EndIf
\State $ctr \gets ctr + 1$
\EndWhile
\State $\mathcal{C} \gets \mathcal{C}B$
\State \Return $\mathcal{C}$
\EndFunction
\end{algorithmic}
\end{algorithm}
	
 For equivalence we prove that external party cannot distinguish by observing a fact of successful block generation whether it is originated from RollerPow or BitcoinPow. We construct following indistinguishability experiment \(POW_{\mathcal{A}}\):

\begin{itemize}
\item There are two honest miners and an adversary. Miners are tossing a fair coin before the experiment. Based on an uniform coin tossing outcome \(b\), one of miners is trying to extend a Rollerchain-based blockchain, another is working on a Bitcoin-based blockchain. The adversary doesn't know about their jobs. Both miners share the same \(q\) and \(D\) values.
\item The adversary \(\mathcal{A}\) generates two valid blockchains, one is Rollerchain and another is Bitcoin, and send both chains to both miners.
\item Both miners are calling their Proof-of-Work functions for an appropriate chain~(ignoring another one). A miner sends \textit{success} if a new block has been successfully generated or \textit{failure} if not, to the adversary. The answer is to be sent not immediately, but just after fixed delay from getting a job \(T_{exp}\).
\item Adversary outputs \(b^\prime\) and succeeds if \(b = b^\prime\). We write \(POW_{\mathcal{A}} = 1\) if she succeeds, and \(POW_{\mathcal{A}} = 0\) otherwise.
\end{itemize}

\begin{assumption}
\label{assume_q}
It should be always possible to perform all the \(q\) iterations for both RollerPow and BitcoinPow functions within \(T_{exp}\) for any of the miners and for any input.
\end{assumption}

\begin{theorem}
For all the PPT adversary \(\mathcal{A}\), \(Pr[POW_{\mathcal{A}} = 1] = \frac{1}{2}\).
\end{theorem}
\begin{proof}
For both Proof-of-Work functions \(hash(\cdot) < D\) check is reached on each call. Let \(\{0,1\}^\mu\) be the range of \(hash(\cdot)\) output, then both miners achieve \textit{success} with the same probability \(p = \frac{D \cdot q}{2^\mu}\). As the adversary does not have an auxiliary information on which result was generated faster~(it is probably from a Bitcoin miner), it cannot distinguish with probability above the random guessing.
\end{proof}

Note that without the Assumption~\ref{assume_q} an adversary can try to generate so big state it is not possible anymore for a RollerPow miner to call \textit{GenTicket} \(q\) times. Then the miner calls \textit{GenTicket} \(g < q\) times, her probability to succeed becomes \(p = \frac{D \cdot g}{2^\mu} \) making possible for the adversary to improve her chance above $\frac{1}{2}$.

\begin{corollary}
\label{cor_gkl}
Assume the GKL~\cite{garay2015bitcoin} environment described in the section~\ref{gkl_general} where miners are running whether \textit{RollerPow} or \textit{BitcoinPow} function, and each miner can perform $q$ iterations for any of both Proof-of-Work functions. Also, the environment $\mathcal{Z}$ does not know which Proof-of-Work function actually being using~(say, messages are encrypted with a secret key the miners share). Then the environment can not distinguish which Proof-of-Work function is used with a probability above the random guessing. That means all the analysis in the original work~(which is based on events probabilities analysis) could be reused. 
\end{corollary}

\subsubsection{Bootstrapping Security}

Consider a new node is connecting to the network. Assume it connects to a large number of nodes, so it sees last \(n\) state snapshots and full blocks. The node is going to download state snapshot from \(n\) blocks ago and then apply the full blocks to it. How secure is the bootstrapping scenario in comparison with a classic one where a full node is downloading and applying blocks since genesis? We introduce notions of full verifier and light verifier and then prove the divergence in bootstrapping results for them is going down exponentially with \(n\).

\begin{definition}A full verifier is a node which downloads all the full blocks since genesis. \end{definition}

\begin{definition}A light verifier is a node downloading only block headers since genesis and then state snapshot from \(n\) blocks ago and last \(n\) full blocks in order to apply the blocks to the state. \end{definition}

\begin{theorem}
Considering an authenticated dictionary second-preimage secure, a full verifier is getting the same state for the same chain as the lite verifier with probability \(1 - negl(n)\).
\end{theorem}
\begin{proof}
Assume the chain has length \(h\). If \(h \le n\) then the both verifiers are applying the same full blocks to the same genesis state, and so get the same results with probability \(1\). Consider the case \(h > n \). As authenticated dictionary is second-preimage secure, a state a light verifier starts from is the same for both verifiers. By applying \(n\) full blocks to it both verifiers is getting the same results if no fork deeper than \(n\) occurs~(it is impossible to get a state and blocks from deeper than \(n\) blocks ago for a light verifier so it will fail to construct the result). Assuming the probability of such a fork is \(p_f\), verifiers are getting the same results with probability \(1 - p_f\). 
Rollback for \(n\) blocks means the common prefix property got broken. By Theorem 9 from~\cite{garay2015bitcoin}~(we can use it due to the Corollary~\ref{cor_gkl} result) the probability of that is at most \(e^{-\Omega(\delta^3 \cdot n)}\), \(\delta \in (0,1) \). Thus \( p_f \le e^{-\Omega(\delta^3 \cdot n)}\) is \(negl(n)\) and both verifiers are getting the same results with probability \(1 - negl(n)\).
\end{proof}

However, a new node can start with a newer state than \(n\) blocks ago. But in case of a fork deeper than the state she will need to download an older state.

\subsection{The \(\varepsilon\)-consensus Attacks}

To perform an \(\varepsilon\)-consensus attack~\cite{luu2015demystifying} a miner includes heavy transactions into the block generated by herself (for free) in order to force other miners to spend some time to verify a block or to skip validation. For our protocol such attacks could be more harmful as a miner needs to apply at least \(k\) blocks, instead of just one. However, a mining software could perform mining states updates in advance, in this case \(\varepsilon\)-consensus attacks are no more harmful than in other systems.

\subsection{Archiving Guarantees In Finite-Size Networks}

Previously we stated that a very large network is storing last \(n\) state snapshots and full blocks. But real networks are limited in size, and the size could be not so big. Consider Alice is joining a network holding a blockchain at height \(h_c\) and sees \(p\) peers, all of them are rational miners. How many versions of state, and also full blocks could Alice find by asking her peers? 

For simplicity, we ignore the fact that the \textit{ChooseSnapshots} values unique for a peer. Then \(p\) peers have \(p \cdot k\) (possibly duplicate) values over \(n\) integer values. By using order statistics the expected minimum value is \( h_{min} = (h_c - n) + \frac{n}{p\cdot k+1} = h_c - \frac{p\cdot k\cdot n}{p\cdot k+1}\). Alice can download state snapshots and blocks since \(h_{min}\) on average from at least one peer.

\section{Related Work}

\subsection{The Permacoin Cryptocurrency}

The key idea behind Permacoin~\cite{miller2014permacoin} is to make the mining process dependent upon the storage resources rather than the computational capabilities. In order to achieve the condition that every miner stores a random subset of a known static dataset, Permacoin requires that the dataset is generated by a trusted dealer. 

The trusted dealer also announces the root hash value of a Merkle tree~\cite{merkle1987digital} built on the top of the dataset. A miner iterates over some random nonce value to find chunks whose indexes are dependent on the value. A ticket contains the nonce value and the chunks with Merkle paths for them being attached. 

The winning ticket gives the right to generate the next block  to be included into the blockchain.

During careful analysis we have found few possible threats and open questions in the protocol:

\begin{itemize}

\item The Permacoin consensus requires a trusted dealer to encode a huge dataset and to build a Merkle tree on top of it. We would like to eliminate the eliminate the need for a trusted dealer completely.

\item The information a miner stores is static because it is a subset of a static dataset, thus it is possible to put it into a dynamic or static RAM and then connect specialized hardware to the memory. With such a scheme companies could get the same advantage over individual miners as in Bitcoin with the same degree of centralization we would like to avoid.

\item Identifiers of data segments to store depend on a miner's public key. Possible identifier collisions are not prohibited by the protocol. So a miner could iterate over private keys in order to find a public key maximizing a number of the collisions (see ``Setup'' formula in ``A simple POR lottery'' figure of the Permacoin paper~\cite{miller2014permacoin}). This could give even more advantage to a company with vast computational resources over individual miners.

\end{itemize}

\subsection{Ethereum}

Ethereum~\cite{ethyp} has state representation as an authenticated data structure being fixed by the protocol with a root hash to be included into a block. However, as there is no incentive to store past state snapshots, the state proof in a block header is useful mostly to help light clients to get the state elements from the full nodes along with the authenticity proofs and does not solve the rationale problem stated in the Section~\ref{rationale}.

Ethereum is contending with specialized hardware by using Ethash algorithm, which is proposing to use pseudo-random dataset generated from blockheaders in the past in a Proof-of-Work mining process. The usefulness of EthHash is the same as of the Bitcoin scheme. The disadvantage of the algorithm is a heavy validation process.

\subsection{Cryptonite}

Cryptonite~\cite{cryptonite} utilizes mini-blockchain scheme~\cite{minibc} to prune full blocks before a constant numbers of last ones, thus having a \textit{proof chain} of length \(h_{pc}\) (the chain of block headers only), a state for a height \(h_{pc}\) and full blocks for heights from \(h_{pc}+1\) till current. It is not clear from the paper~\cite{minibc} how the protocol enforces this scheme.

\section{Further Work}

We highlight some unsolved problems and questions to be answered before a deploy to a real-life system.

\subsection{Compress Block Headers Storage}

It is probably possible to reduce number of block headers stored at a full node by using Proofs-of-Proof-of-Work with sublinear complexity technique from~\cite{kiayiasproofs}. However, the technique should be developed further with respect to dynamic nature of the difficulty parameter.

\subsection{Protection Against Specialized Hardware}

As a set of states needed for mining is dynamic, it is probably more protected from specialized hardware than Permacoin. For better protection asymmetric Proof-of-Work based on the Generalized Birthday problem~\cite{biryukov2016asymmetric} could be incorporated into our protocol. Protection against specialized hardware is needed to prevent a situation when only few parties hold state snapshots and full blocks and so should be considered carefully.

\subsection{A State Exchange Protocol}

Consider Alice joining a network. She chooses a state snapshot at height \(h_d\) to download from her peers. The problem is, if a new block arrives during downloading the snapshot, her peers are to replace the snapshot with another one of height \(h_d+1\). To avoid this problem Alice needs to ask her peers to store the snapshot even if it is not needed for mining anymore. Alice can propose a reward to her peers for doing this, thus some fair protocol is needed.

\subsection{Implementation Parameters and Details}

Archiving parameters \(n\) and \(k\) should be carefully chosen in practice. Also, an authenticated dictionary implementation must have efficient batch updates in order to minimize computational overhead for re-calculating authenticating value during a block application.

\section{Conclusion}
\label{conclusion}

We argue Bitcoin as well as other blockchain systems are secure in some aspects because of altruistic behavior of participants, and the cost of the altruism is going up with time. As we cannot expect such a status quo to be viable in the long term, the question of a blockchain system safe in the presence of rational super-majority arises. In particular we need to protect blockchain systems against a long-term threat of full blocks pruning. A rational full node prunes its full blocks, but if all the nodes are rational then a new node can not bootstrap in a trustless manner. 

We have presented a modification for a Proof-of-Work blockchain system to repurpose the work securing in order to store a fixed number of state snapshots and full blocks collectively. Or result is generic so could be applied to many blockchain systems with different transactional semantics. We have carefully described changes needed to be applied to transactional layer of a blockchain system and a modified Proof-of-Work scheme. 

Together this twofold contribution, the Rollerchain framework, allows full nodes to be rational without a security loss for a network. A rational full node implementation can reduce its storage requirements by three orders of magnitude in comparison with an altruistic one, and a new node still can bootstrap in a trustless, quick and safe way.

\ifCLASSOPTIONcompsoc
  \section*{Acknowledgments}
\else
  \section*{Acknowledgment}
\fi

The authors would like to thank Andrew Miller for a discussion on Permacoin, Roman Oliynykov and Bill White for proofreading.

\bibliographystyle{IEEEtran}
\bibliography{sources.bib}

\begin{thebibliography}{10}
\providecommand{\url}[1]{#1}
\csname url@samestyle\endcsname
\providecommand{\newblock}{\relax}
\providecommand{\bibinfo}[2]{#2}
\providecommand{\BIBentrySTDinterwordspacing}{\spaceskip=0pt\relax}
\providecommand{\BIBentryALTinterwordstretchfactor}{4}
\providecommand{\BIBentryALTinterwordspacing}{\spaceskip=\fontdimen2\font plus
\BIBentryALTinterwordstretchfactor\fontdimen3\font minus
  \fontdimen4\font\relax}
\providecommand{\BIBforeignlanguage}[2]{{%
\expandafter\ifx\csname l@#1\endcsname\relax
\typeout{** WARNING: IEEEtran.bst: No hyphenation pattern has been}%
\typeout{** loaded for the language `#1'. Using the pattern for}%
\typeout{** the default language instead.}%
\else
\language=\csname l@#1\endcsname
\fi
#2}}
\providecommand{\BIBdecl}{\relax}
\BIBdecl

\bibitem{Nakamoto2008}
\BIBentryALTinterwordspacing
S.~Nakamoto, ``{Bitcoin: A Peer-to-Peer Electronic Cash System},'' pp. 1--9,
  2008. [Online]. Available: \url{https://bitcoin.org/bitcoin.pdf}
\BIBentrySTDinterwordspacing

\bibitem{garay2015bitcoin}
J.~Garay, A.~Kiayias, and N.~Leonardos, ``The bitcoin backbone protocol:
  Analysis and applications,'' in \emph{Advances in Cryptology-EUROCRYPT
  2015}.\hskip 1em plus 0.5em minus 0.4em\relax Springer, 2015, pp. 281--310.

\bibitem{milleranonymous}
\BIBentryALTinterwordspacing
A.~Miller and J.~J. LaViola~Jr, ``Anonymous byzantine consensus from
  moderately-hard puzzles: A model for bitcoin.'' [Online]. Available:
  \url{https://socrates1024.s3.amazonaws.com/consensus.pdf}
\BIBentrySTDinterwordspacing

\bibitem{cromanscaling}
\BIBentryALTinterwordspacing
K.~Croman, C.~Decker, I.~Eyal, A.~E. Gencer, A.~Juels, A.~Kosba, A.~Miller,
  P.~Saxena, E.~Shi, and E.~G{\"u}n, ``On scaling decentralized blockchains.''
  [Online]. Available: \url{http://fc16.ifca.ai/bitcoin/papers/CDE+16.pdf}
\BIBentrySTDinterwordspacing

\bibitem{douceur2002sybil}
J.~R. Douceur, ``The sybil attack,'' in \emph{Peer-to-peer Systems}.\hskip 1em
  plus 0.5em minus 0.4em\relax Springer, 2002, pp. 251--260.

\bibitem{miller2014permacoin}
A.~Miller, A.~Juels, E.~Shi, B.~Parno, and J.~Katz, ``Permacoin: Repurposing
  bitcoin work for data preservation,'' in \emph{Security and Privacy (SP),
  2014 IEEE Symposium on}.\hskip 1em plus 0.5em minus 0.4em\relax IEEE, 2014,
  pp. 475--490.

\bibitem{bitcoind}
\BIBentryALTinterwordspacing
``{Bitcoin Core Code Repository}.'' [Online]. Available:
  \url{https://github.com/bitcoin/bitcoin/}
\BIBentrySTDinterwordspacing

\bibitem{hardin1968commons}
\BIBentryALTinterwordspacing
G.~Hardin, ``The tragedy of the commons,'' \emph{Science}, vol. 162, pp.
  1243--1248, December 1968. [Online]. Available:
  \url{http://www.sciencemag.org/cgi/reprint/162/3859/1243.pdf}
\BIBentrySTDinterwordspacing

\bibitem{snapshot}
\BIBentryALTinterwordspacing
``{Bitcoin Blockchain Data Torrent}.'' [Online]. Available:
  \url{https://bitcointalk.org/index.php?topic=145386.0}
\BIBentrySTDinterwordspacing

\bibitem{Canetti:2001}
\BIBentryALTinterwordspacing
R.~Canetti, ``Universally composable security: A new paradigm for cryptographic
  protocols,'' pp. 136--, 2001. [Online]. Available:
  \url{http://dl.acm.org/citation.cfm?id=874063.875553}
\BIBentrySTDinterwordspacing

\bibitem{utxoset}
\BIBentryALTinterwordspacing
A.~Miller, ``Storing utxos in a balanced merkle tree,'' aug 2012. [Online].
  Available: \url{https://bitcointalk.org/index.php?topic=101734.0}
\BIBentrySTDinterwordspacing

\bibitem{ethyp}
\BIBentryALTinterwordspacing
``{Ethereum: A Secure Decentralized Generalized Transaction Ledger}.''
  [Online]. Available: \url{http://gavwood.com/Paper.pdf}
\BIBentrySTDinterwordspacing

\bibitem{nissim1998certificate}
K.~Nissim and M.~Naor, ``Certificate revocation and certificate update.'' in
  \emph{USENIX Security}.\hskip 1em plus 0.5em minus 0.4em\relax Citeseer,
  1998.

\bibitem{kocher1998certificate}
P.~C. Kocher, ``On certificate revocation and validation,'' in
  \emph{International Conference on Financial Cryptography}.\hskip 1em plus
  0.5em minus 0.4em\relax Springer, 1998, pp. 172--177.

\bibitem{sparsemerkle}
R.~Dahlberg, T.~Pulls, and R.~Peeters, ``Efficient sparse merkle trees: Caching
  strategies and secure (non-)membership proofs,'' Cryptology ePrint Archive,
  Report 2016/683, 2016, \url{http://eprint.iacr.org/2016/683}.

\bibitem{crosby2011authenticated}
S.~A. Crosby and D.~S. Wallach, ``Authenticated dictionaries: Real-world costs
  and trade-offs,'' \emph{ACM Transactions on Information and System Security
  (TISSEC)}, vol.~14, no.~2, p.~17, 2011.

\bibitem{anagnostopoulos2001persistent}
A.~Anagnostopoulos, M.~T. Goodrich, and R.~Tamassia, ``Persistent authenticated
  dictionaries and their applications,'' in \emph{International Conference on
  Information Security}.\hskip 1em plus 0.5em minus 0.4em\relax Springer, 2001,
  pp. 379--393.

\bibitem{luu2015demystifying}
L.~Luu, J.~Teutsch, R.~Kulkarni, and P.~Saxena, ``Demystifying incentives in
  the consensus computer,'' in \emph{Proceedings of the 22nd ACM SIGSAC
  Conference on Computer and Communications Security}.\hskip 1em plus 0.5em
  minus 0.4em\relax ACM, 2015, pp. 706--719.

\bibitem{merkle1987digital}
R.~C. Merkle, ``A digital signature based on a conventional encryption
  function,'' in \emph{Advances in Cryptology—CRYPTO’87}.\hskip 1em plus
  0.5em minus 0.4em\relax Springer, 1987, pp. 369--378.

\bibitem{cryptonite}
\BIBentryALTinterwordspacing
``{The Cryptonite Project Homepage}.'' [Online]. Available:
  \url{http://cryptonite.info}
\BIBentrySTDinterwordspacing

\bibitem{minibc}
\BIBentryALTinterwordspacing
``{The Mini-Blockchain Scheme}.'' [Online]. Available:
  \url{http://cryptonite.info/files/mbc-scheme-rev2.pdf}
\BIBentrySTDinterwordspacing

\bibitem{kiayiasproofs}
A.~Kiayias, N.~Lamprou, and A.-P. Stouka, ``Proofs of proofs of work with
  sublinear complexity.''

\bibitem{biryukov2016asymmetric}
A.~Biryukov and D.~Khovratovich, ``Asymmetric proof-of-work based on the
  generalized birthday problem,'' \emph{Proceedings of NDSS 2016}, p.~13, 2016.

\end{thebibliography}

\newpage
\onecolumn

\end{document}